\documentclass{article}

\usepackage{amsmath, amsthm, amssymb}
\usepackage{graphicx}
\usepackage{enumerate}

\newtheorem{theorem}{Theorem}
\newtheorem{lemma}[theorem]{Lemma}
\newtheorem{proposition}[theorem]{Proposition}

\theoremstyle{remark}
\newtheorem{remark}[theorem]{Remark}

\theoremstyle{definition}
\newtheorem{property}{Property}

\newcommand{\tas}{\ensuremath{\mathcal{T}}}

\newcommand{\dtilec}[1]{\ensuremath{\mathrm{C}^{\mathrm{dtilec}(#1)}}}

\newcommand{\findstrength}{\ensuremath{\mbox{\sc FindStrength}}}
\newcommand{\findoptstrength}{\ensuremath{\mbox{\sc FindOptimalStrength}}}
\newcommand{\SAT}{\ensuremath{\mbox{\sc Sat}}}
\newcommand{\oneinthree}{\ensuremath{\mbox{\sc 1-in-3-\SAT}}}
\newcommand{\TP}{\ensuremath{\mbox{\sc ThresholdProgramming}}}

\newcommand{\NP}{{\tt NP}}
\newcommand{\north}{{\tt N}}
\newcommand{\west}{{\tt W}}
\newcommand{\south}{{\tt S}}
\newcommand{\east}{{\tt E}}
\newcommand{\dom}{\ensuremath{\mathrm{dom}}}
\renewcommand{\square}[1]{\ensuremath{\mathit{Sq}_{#1}}}
\renewcommand{\Vec}[1]{\ensuremath{\mbox{\boldmath$#1$}}}

\title{On the Behavior of Tile Assembly System at High Temperatures\footnote{
This paper is an extended version of \cite{SekiOkuno2012}. 
}}
\author{Shinnosuke Seki and Yasushi Okuno}

\begin{document}

\maketitle

\begin{abstract}
	Behaviors of Winfree's tile assembly systems (TASs) at high temperatures are investigated in combination with integer programming of a specific form called threshold programming. 
	First, we propose a way to build bridges from the Boolean satisfiability problem ($\SAT$) to threshold programming, and further to TAS's behavior, in order to prove the \NP-hardness of optimizing temperatures of TASs that behave in a way given as input. 
	These bridges will take us further to two important results on the behavior of TASs at high temperatures. 
	The first says that arbitrarily high temperatures are required to assemble some shape by a TAS of ``reasonable'' size. 
	The second is that for any temperature $\tau \ge 4$ given as a parameter, it is \NP-hard to find the minimum size TAS that self-assembles a given shape and works at a temperature below $\tau$.   
\end{abstract}

	\section{Introduction}

The abstract Tile Assembly Model (aTAM), which has been introduced by Winfree \cite{Winfree1998} based on a dynamic version of Wang tiling \cite{Wang1961}, is a model of ``programmable crystal growth'' with algorithmically-rich theoretical background and results. 
Self-assembling systems in this model are called {\it tile assembly systems} (TASs). 
Self-assembling (square) tiles have been experimentally implemented as DNA double-crossover molecules in 1998 \cite{WiLiWeSe1998}, which are designed so ingeniously that they bind deterministically into a single target shape, even subject to the chaotic nature of molecules floating randomly in a well-mixed solution. 
The last three decades have seen drastic advancements on the reliability of DNA tile assembly; in fact, the error rate of 10\% per tile in 2004 was improved down to 0.13\% in 2009 \cite{BaScRoWi2009}. 

In the aTAM, sticky ends are abstracted to be a glue label, and their strengths are assigned by a strength function $g$. 
A square tile adheres stably to an aggregate of tiles whenever the sum of the strengths of neighboring sides with matching labels according to $g$ exceeds a threshold $\tau$, which is a system parameter called {\it temperature}. 
``Temperature'' in aTAM is rather metaphorical than actually describing the physical temperature of the experimental environment. 
It nonetheless can be interpreted as a physical metric for the granularity with which different energy levels must be distinguished in order for TASs to behave as expected. 
Certain ``behaviors'' of TASs were proved to require three different energy levels such that the gap between two of them is exponentially larger than the gap between other two \cite{ChenDotySeki2011} (for a formal definition of TAS's behavior as {\it strength-free TAS}, see Section~\ref{sec:behavior}). 
Conventionally, the glue strengths (and hence, temperature) of TASs are assumed to be integers, without loss of generality. 
Indeed, this is a way of ``quantizing'' the minimum distinction we are willing to make between energies and then re-scalling so that this quantity is normalized to 1.

As mentioned briefly above, the stability of the attachment of a tile at a position is determined by the sum of the strengths that tiles at the neighboring positions offer via their abutting edges, relative to $\tau$. 
In aTAM, the basic building blocks, that is, tiles, are abstracted to be the $1 \times 1$ square so that a tile can be adjacent to at most four other tiles. 
Thus, the sum that determines the attachment stability consists of at most four terms. 
This motivates us to study a system of inequalities whose left-hand side consists of at most 4 terms (non-negative integers or constants) and whose right-hand side is $\tau$. 
We call such an inequality a {\it $\tau$-inequality of at most 4 terms}; we use this term often with the replacement of $\tau$ by either $\ge_\tau$ (at least $\tau$) or $<_\tau$ (strictly less than $\tau$) to specify its sign. 
Then, we can say that $\tau$-inequalities of at most 4 terms dominate the behavior of a TAS at the micro, or {\it local}, level, that is, per attachment event. 
Optimizing (minimizing) $\tau$ under $\tau$-inequalities is a specific type of integer programming we call {\it threshold programming} (TP). 
In this paper, we will prove that TP is \NP-hard even under the condition that all $\ge_\tau$-inequalities involved be of at most 4 terms and all $<_\tau$-inequalities involved be of at most 3 terms (Lemma~\ref{lem:TP43_NPhard}). 
This condition makes TPs reducible to the local behavior of a TAS, and this implies the \NP-hardness of the problem $\findoptstrength$, which aims at optimizing the temperature of TASs that behave in a way specified as input (Theorem~\ref{thm:findoptstrength_NPhard}). 
In other words, $\findoptstrength$ cannot be solved in a polynomial time, unless ${\tt P} = \NP$. 
This \NP-hardness complements a result by Chen, Doty, and Seki \cite{ChenDotySeki2011}, which affirmatively answered a problem posed by Adleman et al.~\cite{AdChGoHuKeEsRo2002}. 

The TP instances obtained in this reduction will lead us further to the study of macro, terminal, or {\it global}, behaviors of TASs. 
The global behavior simply concerns what shape(s) a TAS assembles, and does not mind how its underlying attachment events proceed. 
Well-examined problems on the global behavior of TASs includes finding the optimal design of TASs that certainly assembles the $n \times n$ square for a given $n$ \cite{AdChGoHu2001, AdChGoHuKeEsRo2002, RothemundWinfree2000}. 
There and also in this paper, the optimality of TASs is measured by the number of distinct tile types they contain, and this criterion is called the {\it tile complexity}.  
The most important contribution we make in this paper along this line is the proposal of a unified framework to convert a system $\mathcal{S}$ of $\tau$-inequalities into a shape $S$ with the property that if $\mathcal{S}$ is solvable for $\tau = k$, then the resulting shape $S$ can be assembled by a temperature-$k$ TAS of ``reasonable size'' (Property~\ref{property:TAS_as_inequality_system_solver}). 
According to this framework, it suffices to choose a system of $\tau$-inequalities that requires $\tau \ge k$ for its solvability in order to obtain a shape which prefers the temperature $k$ to those below with respect to tile complexity (Theorem~\ref{thm:tau_better_than_less}). 
A choice of another system proves to enable this framework prove also that, for any $\tau \ge 4$, it is \NP-hard to compute the minimum number of tile types required for TASs at a temperature at most $\tau$ to self-assemble the shape (Theorem~\ref{thm:TC_above_4_NPhard}). 


Current laboratory techniques allow us to handle only at most 2 distinct energy levels, that is, temperature 2 (even making a distinction between two energy levels is difficult, see, e.g., \cite{CookFuSchweller2011} and references therein, but successful designs of self-assembling molecular systems at temperature 2 have been reported \cite{FuHaPaWiMu2009,RothemundPapadakisWinfree2004}). 
Therefore, as of this date, we cannot help but interpret our results as computational infeasibility of determining whether behaviors of TAS can be implemented physically. 

This paper is organized as follows. 
After the preliminary section, in Section~\ref{sec:behavior}, we will formalize the local behavioral equivalence among TASs. 
The section consists of opening paragraphs that introduce the formal definition of the equivalence and formalize the problem $\findoptstrength$, being followed by two subsections: 
Section~\ref{subsec:TP_preliminaries} is a preliminary for threshold programming and makes necessary preparations for the proof of the \NP-hardness of $\findoptstrength$. 
Then in Section~\ref{subsec:findoptstrength_NPhard} we present the succession of \NP-hardness results: quadripartite $\oneinthree$ (Lemma~\ref{lem:quad_1in3SAT}), threshold programming (Lemma~\ref{lem:TP43_NPhard}), and then $\findoptstrength$ (Theorem~\ref{thm:findoptstrength_NPhard}), chained by Karp-reductions. 
Using these results, in Section~\ref{sec:complexity}, we will prove the above-mentioned theorems on the global behavior of TASs. 

	\section{Abstract Tile Assembly Model}

This section aims at tersely introducing the reader to the aTAM.
An excellent tutorial can be found, for instance, in \cite{RothemundWinfree2000}. 

Let $\Sigma$ be an alphabet, and by $\Sigma^*$, we denote the set of finite strings over $\Sigma$. 
By $\mathbb{Z}$ and $\mathbb{N}$, we denote the set of integers and the set of positive integers, respectively, and let $\mathbb{N}_0 = \mathbb{N} \cup \{0\}$. 
In aTAM, $\mathbb{Z}^2$ is especially considered either as the two-dimensional integer lattice or as the set of all points on it. 

Given a set of points $A \subseteq \mathbb{Z}^2$ on the integer lattice, the {\it full grid graph} of $A$ is the undirected graph $G^{\rm f}_A = (V, E)$, where $V = A$ and for all $u, v \in V$, there is an edge between $u$ and $v$ if and only if $||u-v||_2 = 1$, where $|| \cdot ||_2$ is the Manhattan distance, that is, $u$ and $v$ are adjacent points. 
A {\it shape} is a set $S \subseteq \mathbb{Z}^2$ such that $G^{\rm f}_S$ is connected, and we denote the set of all {\it finite} shapes by $\mathcal{FS} \subseteq \mathcal{P}(\mathbb{Z}^2)$.
Let $\north, \west, \south, \east$ stand for the respective directions north, west, south, and east, and be also interpreted as the respective unit vectors $(0, 1), (-1, 0), (0, -1), (1, 0)$. 

A {\it tile type} $t$ is a quadruple $t \in \Sigma^* \times \Sigma^* \times \Sigma^* \times \Sigma^*$, and is regarded as a unit square with four sides listed in the counter-clockwise order starting at the north (\north), each having a {\it glue label} (a.k.a., {\it glue}) taken from $\Sigma^*$. 
For each direction $d \in \{\north, \west, \south, \east\}$, let $t(d)$ be the glue label at the $d$ side of $t$. 
Let $T$ be a {\it finite} set of tile types, and let us denote the (finite) set of all glues of tile types in $T$ by $\Lambda(T) \subsetneq \Sigma^*$. 
An {\it assembly} (a.k.a., {\it supertile}) is a positioning of tiles of types in $T$ on (part of) the integer lattice $\mathbb{Z}^2$.  
It does not have to be a tessellation. 
Hence, we can say that an assembly is a partial function $\mathbb{Z}^2 \dashrightarrow T$.
Given two assemblies $\alpha, \beta: \mathbb{Z}^2 \dashrightarrow T$, $\alpha$ is a {\it sub-assembly} of $\beta$, written as $\alpha \sqsubseteq \beta$, if $\dom(\alpha) \subseteq \dom(\beta)$ and for every point $p \in \dom(\alpha)$, $\alpha(p) = \beta(p)$, where $\dom$ denotes the domain of the function. 

The aTAM models dynamics in the growth of assemblies based on the interaction among its basic building blocks, tiles. 
A {\it strength function $g: \Lambda(T) \to \mathbb{N}_0$} endows tiles with an ability to interact with its neighboring tiles by assigning the strength $g(\ell)$ to the {\it matching} label $\ell$ of their abutting edges.
If the labels do not match or $g(\ell) = 0$, these tiles do not interact; otherwise, they do. 
An assembly $\alpha$ and a strength function $g$ induce a {\it binding graph}, which is a grid graph whose vertices are $\dom(\alpha)$ and for two neighboring positions $p_1, p_2 \in \dom(\alpha)$, there is an edge between $p_1$ and $p_2$ on this graph if and only if the tiles $\alpha(p_1)$ and $\alpha(p_2)$ interact.  
On this graph, an edge between vertices means that the corresponding tiles interact, and hence, their abutting edges share the same label $\ell$. 
Thus, we can consider that the edge is labeled with $\ell$ and $g$ gives it the weight $g(\ell)$. 
The assembly is {\it $\tau$-stable (with respect to $g$)} if every cut of its binding graph has strength at least $\tau$. 
That is, the assembly is $\tau$-stable if at least energy $\tau$ is required to separate it into two parts. 

A {\it (seeded) tile assembly system} (TAS) is a quadruple $\tas = (T, \sigma, g, \tau)$, where $T$ and $g$ are as stated above, $\sigma: \mathbb{Z}^2 \dashrightarrow T$ is a finite $\tau$-stable {\it seed assembly}, and $\tau \ge 1$ is an integer parameter called {\it temperature}. 
TASs are provided with inexhaustible supply of copies of each tile type, each copy being referred to as a {\it tile}. 
If the seed assembly $\sigma$ consists of a single tile, then $\tas$ is said to be {\it singly-seeded}. 
In this paper, we consider only singly-seeded TASs.

Given two $\tau$-stable assemblies $\alpha, \beta$, we write $\alpha \to_1^{\tas} \beta$ if $\alpha \sqsubseteq \beta$ and $\dom(\beta) \setminus \dom(\alpha) = \{p\}$ for some position $p \in \mathbb{Z}^2$. 
Intuitively, this means that $\alpha$ can grow into $\beta$ by the addition of a single tile at the position $p$. 
Since $\beta$ is required to be $\tau$-stable, the new tile is able to bind to $\alpha$ with strength at least $\tau$. 
In this case, we say that {\it $\alpha$ $\tas$-produces $\beta$ in one step}. 

A sequence of $\tau$-stable assemblies $\alpha_0, \alpha_1, \ldots, \alpha_k$ is a {\it $\tas$-assembly sequence} if for all $1 \le i \le k$, $\alpha_{i-1} \to_1^{\tas} \alpha_i$ holds. 
We write $\alpha \to^\tas \beta$ and say {\it $\alpha$ $\tas$-produces $\beta$} (in 0 or more steps) if there is a $\tas$-assembly sequence $\alpha_0, \alpha_1, \ldots, \alpha_k$ of length $k = |\dom(\beta) \setminus \dom(\alpha)|$ with $\alpha_0 = \alpha$ and $\alpha_k = \beta$. 
(This definition of producibility is justified by our limited focus only onto the finite assemblies in this paper; for the infinite assembly, it is not appropriate; see \cite{BrChDoKaSe2011} for instance.) 
An assembly $\alpha$ is {\it $\tas$-producible} or {\it producible by $\tas$} if $\sigma \to^\tas \alpha$. 
A $\tau$-stable assembly $\alpha$ is {\it ($\tas$-)terminal} if for any $\tau$-stable assembly $\beta$, $\alpha \to^\tas \beta$ implies $\alpha = \beta$. 
Let $\mathcal{A}[\tas]$ be the set of assemblies producible by $\tas$, and let $\mathcal{A}_\Box[\tas] \subseteq \mathcal{A}[\tas]$ be the set of terminal assemblies that are producible by $\tas$. 
A TAS $\tas$ is {\it directed} if the poset $(\mathcal{A}[\tas], \to^\tas)$ is directed, i.e., for each $\alpha, \beta \in \mathcal{A}[\tas]$, there exists $\gamma \in \mathcal{A}[\tas]$ such that $\alpha \to^\tas \gamma$ and $\beta \to^\tas \gamma$. 
Given a shape $S \subseteq \mathbb{Z}^2$, a TAS $\tas$ {\it strictly (a.k.a.,~uniquely) (self-)assembles} $S$ if the shape of every terminal assembly produced by $\tas$ is $S$. 

	\subsection{Directed tile complexity}

A directed TAS that strictly self-assembles a shape $S$ can be regarded as a ``program'' to output the shape. 
A measure of how concisely one can describe such a TAS with respect to the number of tile types was introduced by Rothemund and Winfree \cite{RothemundWinfree2000} in the name of directed tile complexity of $S$. 
This complexity has been intensely investigated for TASs at the temperature 1 or 2 \cite{AdChGoHuKeEsRo2002}. 
The {\it temperature-2 directed tile complexity} of $S$ is formally defined as follows: 
\[
	\dtilec{2}(S) = \min \biggl\{|T| 
	\left| 
	\begin{array}{r}
	\mbox{$\tas = (T, \sigma, g, 2)$ is a directed TAS} \\ 
	\mbox{that strictly self-assembles $S$} 
	\end{array}
	\right\}. 
\]
This is the minimum number of tile types required for a directed TAS at the temperature 2 to strictly self-assemble $S$. 
Since any temperature-1 TAS can be simulated at the temperature 2 simply by doubling the strength associated to each of its labels (see also Proposition~\ref{prop:simulation_at_high_temp}), this definition indeed has already taken the temperature-1 TASs into account. 
We parameterize this complexity measure by a temperature $\tau$ as: 
\[
	\dtilec{\tau}(S) = \min \biggl\{|T| 
	\left| 
	\begin{array}{r}
	\mbox{$\tas = (T, \sigma, g, \tau)$ is a directed TAS} \\ 
	\mbox{that strictly self-assembles $S$} 
	\end{array}
	\right\},  
\]
and introduce it as {\it directed tile complexity of $S$ at the temperature $\tau$}. 

As mentioned above, $\dtilec{1}(S) \ge \dtilec{2}(S)$ for any $S$. 
Now we show that for any temperature $\tau$ and a positive integer $k$, $\dtilec{\tau}(S) \ge \dtilec{k\tau}(S)$ holds. 

\begin{proposition}\label{prop:simulation_at_high_temp}
	For any $\tau \in \mathbb{N}$, TASs at the temperature $\tau$ can be simulated at any temperature that is a multiple of $\tau$. 
\end{proposition}
\begin{proof}
	This simulation is simply done by multiplying the strengths and $\tau$ of a given TAS $\tas_1 = (T, \sigma, g, \tau)$ by a proper constant $c$. 
	The TAS thus obtained is $\tas_2 = (T, \sigma, g', c\tau)$ with $g'(\ell) = c g(\ell)$ for each glue label $\ell$ in $T$. 
\end{proof}

Neither this proposition nor Theorem~\ref{thm:tau_better_than_less} in Section~\ref{sec:complexity} implies the monotonically decreasing property of $\dtilec{\tau}(S)$, being considered as a function of $\tau$. 
Indeed, it is not so as being exemplified at the end of Section~\ref{sec:complexity}. 
This non-monotonicity motivates us to introduce the notion of {\it directed tile complexity of $S$ at the temperatures at most $\tau$}. 
This measure is to be defined as 
\[
	\dtilec{\le \tau}(S) = \min\{\dtilec{i}(S) \mid 1 \le i \le \tau\}.
\]

As the TASs $\tas_1$ and $\tas_2$ in the proof of Proposition~\ref{prop:simulation_at_high_temp}, more than one TAS can exhibit identical behaviors at the local (attachment) level in the sense that whenever some of the four sides of a tile cooperate for the stable attachment in one TAS, a tile of the same type do so in the other TASs, though these TAS may be at different temperatures and may assign each label with different strengths. 
As a result, they produce the same shapes. 
In the next section, we will formalize this behavioral equivalence among TASs at the local level, and work problems related to the temperature and the parameterized tile complexity that were left open in \cite{AdChGoHuKeEsRo2002, ChenDotySeki2011}. 

	\section{Behavioral Equivalences among TASs}
	\label{sec:behavior}

The ``behavior'' of a TAS $\tas = (T, \sigma, g, \tau)$ is determined fully by its strength function $g$ and temperature $\tau$. 
More precisely, $g$ and $\tau$ do so by specifying the local behavior of each tile type $t \in T$ in the form of {\it cooperation set of $t$ with respect to $g$ and $\tau$} \cite{ChenDotySeki2011}, which is defined as:  
\[
	\mathcal{D}_{g, \tau}(t) = \Bigl\{ D \subseteq \{\north, \west, \south, \east\} \Bigm| \mbox{$\sum_{d \in D} g(t(d)) \ge \tau$} \Bigr\}. 
\]
This is the collection of sets of four sides of $t$ whose glues have sufficient strengths to bind $t$ cooperatively. 
By definition, if a set of sides of $t$ is in $\mathcal{D}_{g, \tau}(t)$, then any of its superset is also included in $\mathcal{D}_{g, \tau}(t)$. 
Any tile type whose cooperation set is empty can be rid from $\tas$ because a tile of that type never attaches. 
Combining these together, we assume that for any $t \in T$, $\{\north, \west, \south, \east\} \in \mathcal{D}_{g, \tau}(t)$. 

Cooperation sets provide a behavioral equivalence among TASs. 
Given $\tas_1 = (T, \sigma, g_1, \tau_1)$ and $\tas_2 = (T, \sigma, g_2, \tau_2)$ that share the tile set $T$ and seed $\sigma$, if $\mathcal{D}_{g_1, \tau_1}(t) = \mathcal{D}_{g_2, \tau_2}(t)$ for each tile type $t \in T$, then these TASs are said to be {\it locally equivalent} (written as $\tas_1 \sim \tas_2$) \cite{ChenDotySeki2011}. 
The behaviors of locally equivalent TASs are exactly the same at the tile attachment level.
This easily leads us to the property that a sequence of assemblies $\alpha_0, \alpha_1, \ldots, \alpha_k$ is a $\tas_1$-assembly sequence if and only if it is a $\tas_2$-assembly sequence. 
From this it follows that these TASs produce the same assemblies as well as the same terminal assemblies. 
In short, $\tas_1 \sim \tas_2$ implies $\mathcal{A}[\tas_1] = \mathcal{A}[\tas_2]$ and $\mathcal{A}_\Box[\tas_1] = \mathcal{A}_\Box[\tas_2]$. 
As a corollary, we can see that given locally equivalent TASs, one is directed and strictly self-assembles a shape if and only if so is and does the other. 

The local equivalence $\sim$ divides the set of all TASs into the equivalence classes. 
All the TASs in a resulting equivalence class behave exactly in the same way locally, and hence, we are allowed to use the term ``behavior of this class.'' 
This means that the class can choose a TAS $\tas = (T, \sigma, g, \tau)$ it contains arbitrarily as representative in describing its behavior by the pair $(T, \{\mathcal{D}_{g, \tau}(t) \mid t \in T\})$. 
This suggests a way to define a variant of TAS by assigning each $t \in T$ with a set of subsets of $\{\north, \west, \south, \east\}$ as a cooperation set. 
Chen, Doty, and Seki introduced this variant called {\it strength-free TAS} \cite{ChenDotySeki2011} as it is free from strength function or temperature. 
Formally, a {\it strength-free TAS} is a triple $(T, \sigma, \mathcal{D})$, where $T$ and $\sigma$ are the same as those for TAS, while $\mathcal{D}: T \to \mathcal{P}(\mathcal{P}(\{\north, \west, \south, \east\}))$ is a function from a tile type $t \in T$ to a set of subsets of $\{\north, \west, \south, \east\}$ that is closed under superset operation, where $\mathcal{P}$ means the power set. 
As done between TASs, we can define the local equivalence between a TAS $(T, \sigma, g, \tau)$ and a strength-free TAS $\tas_{\rm sf} = (T, \sigma, \mathcal{D})$ as: they are {\it locally-equivalent} if $\mathcal{D}(t) = \mathcal{D}_{g, \tau}(t)$ for each $t \in T$. 
For an equivalence class, if $\tas_{\rm sf}$ is locally equivalent to an element of the class, then so is it to all of them. 
This observation allows us to regard the strength-free TAS as a representative of (the local behavior of) this class. 
Of particular note is that such a strength-free TAS is unique for each class. 
It must be also noted that there exists a strength-free TAS that does not represent any class, that is, that is not locally equivalent to any TAS. 
This implementability check was formalized as $\findstrength$ in \cite{ChenDotySeki2011}, which is defined as follows:
\[
\begin{array}{l}
\findstrength \\
\begin{array}{ll}
{\tt INPUT:}  	& \mbox{a strength-free TAS $\tas_{\rm sf}$} \\
{\tt OUTPUT:} 	& \mbox{a TAS that is locally equivalent to $\tas_{\rm sf}$, if any,}\\  
		& \mbox{or reports that none of such TAS exists otherwise}. 
\end{array}
\end{array}
\]
A polynomial-time algorithm for this problem was proposed in \cite{ChenDotySeki2011}. 

The strength-free TAS was introduced as a technical tool to solve an open problem posed by Adleman et al.~\cite{AdChGoHuKeEsRo2002}. 
In the paper, they proposed an algorithm to find a minimum size directed TAS $\tas = (T, \sigma, g, 2)$ that strictly self-assembles the $n \times n$ square $\square{n}$, subject to the constraint that the TAS's temperature is 2. 
Note that $\dtilec{2}(\square{n}) = O(\frac{\log n}{\log \log n})$ \cite{AdChGoHu2001}. 
This algorithm enumerates all temperature-2 TASs with at most $\dtilec{2}(\square{n})$ tile types, and checks whether each of them is directed and strictly self-assembles $\square{n}$ (this is proved to be polynomial-time checkable in $n$). 
The temperature of a system to be checked need not be 2, but rather the temperature-2 restriction\footnote{This can be replaced with the temperature-$c$ restriction for any constant $c \ge 1$.} is utilized to upper-bound the number of all candidates to be thus checked by a polynomial in $n$, and as a result, this algorithm runs in a polynomial time in $n$. 
The open problem was whether the temperature upper-bound could be removed. 

Chen, Doty, and Seki answered this problem positively by designing an algorithm that runs in polynomial time in $n$ without relying on such an upper-bound \cite{ChenDotySeki2011}. 
Though being based on the above-mentioned idea by Adleman et al.~basically, their algorithm enumerates strength-free TASs with at most $\dtilec{2}(\square{n})$ tile types instead of conventional ones, and commits extra check for the implementability. 
This algorithm raised another problem of whether it could be modified so as to output among all the minimum size directed TASs for $\square{n}$ the one(s) working at the lowest temperature. 
This motivates us to study the following optimization: 
\[
\begin{array}{l}
\findoptstrength \\
\begin{array}{ll}
{\tt INPUT:}  	& \mbox{a strength-free TAS $\tas_{\rm sf}$} \\
{\tt OUTPUT:} 	& \mbox{a TAS of minimal temperature that is locally equivalent to $\tas_{\rm sf}$} \\
		& \mbox{if any, or reports that none of such TAS exists otherwise}. 
\end{array}
\end{array}
\]
One of the main contributions of this paper is to prove the \NP-hardness of this problem (Theorem~\ref{thm:findoptstrength_NPhard}). 

	\subsection{Threshold Programming}
	\label{subsec:TP_preliminaries}

0-1 integer programming, one of the Karp's 21 \NP-complete problems \cite{Karp1972}, is a decision variant of integer programming in which all the variables are restricted to be binary. 
In order to prove the completeness, Karp employed a polynomial-time reduction from the Boolean satisfiability problem ($\SAT$) to this problem. 

In order to prove the \NP-hardness of $\findoptstrength$, let us introduce a subclass of integer programming (IP) that aims at optimizing $\tau$ subject to only $\ge_\tau$-inequalities and $<_\tau$-inequalities. 
We call such an IP a {\it threshold programming} (TP). 
This is formalized as: for given integer matrices $C_1, C_2$, minimize $\tau$ on condition that there exists a nonnegative integer vector $\Vec{x}$ satisfying 
\[
	\mbox{$C_1 \vec{x} \ge \tau \Vec{1}$ ($\ge_\tau$-inequalities) and $C_2 \vec{x} < \tau \Vec{1}$ ($<_\tau$-inequalities)}, 
\]
where $\Vec{1}$ is the vector all of whose components are 1. 

$\findoptstrength$ is actually a TP any of whose constraints is either a $\ge_\tau$-inequality of at most 4 terms or a $<_\tau$-inequality of at most 3 terms. 
In order to see this, let us consider a simple strength-free TAS with only one tile type $t = (\ell_1, \ell_2, \ell_3, \ell_4)$ and a cooperation set $\mathcal{D}(t) = \{\{\north, \west, \south, \east\}\}$.
Finding a TAS at lowest temperature that is locally equivalent to this strength-free TAS is equal to solving the following system of inequalities with a {\it positive integer} variable $\tau$: 
\begin{equation}\label{eq:TP_example_1}
\left\{
\begin{array}{lll}
  \ell_1 + \ell_2 + \ell_3 + \ell_4 &\ge& \tau \\
  \ell_1 + \ell_2 + \ell_3 &<& \tau \\
  \ell_1 + \ell_2 + \ell_4 &<& \tau \\
  \ell_1 + \ell_3 + \ell_4 &<& \tau \\
  \ell_2 + \ell_3 + \ell_4 &<& \tau \\
\end{array}
\right.
\end{equation}
so as to minimize $\tau$. 
The minimum $\tau$ for this system is 4 with $\ell_1 = \ell_2 = \ell_3 = \ell_4 = 1$ (with $\tau$ being strictly less than 4, this system is not solvable). 
The first inequality in \eqref{eq:TP_example_1} corresponds to $\{\north, \west, \south, \east\} \in \mathcal{D}(t)$, the second corresponds to $\{\north, \west, \south\} \not\in \mathcal{D}(t)$, and so forth. 
Though \eqref{eq:TP_example_1} does not seem to contain an inequality corresponding to $\{\north, \west\} \not\in \mathcal{D}(t)$, the second one (or third) actually implies this, and hence, not presented for the sake of space. 
In any case, the system for a cooperation set contains at most 15 inequalities. 
We denote the class of TPs any of whose constraints contains at most the number of terms as specified above by TP(4,~3). 
Its decision variant, denoted by $\tau$-$\TP(4, 3)$ or simply $\tau$-TP(4,~3), is of interest in which $\tau$ is not a variable but a given constant, and one is asked to decide whether $\Vec{x}$ exists. 

Our arguments below mainly consist of designing systems of $\tau$-inequalities for various purposes. 
As a tool, we introduce a sub-system that will be embedded into these systems and force their variables to assume at least or exactly some specific value. 
It is built on the following pair of $\tau$-inequalities: 
\begin{equation}\label{eq:positive}
	x_1 + x_a \ge \tau \mbox{ and } x_a < \tau. 
\end{equation} 
This pair implies $x_1 \ge 1$. 
Once \eqref{eq:positive} being embedded into a $\tau$-inequality system, the variable $x_1$ cannot help but assume a positive value itself ($x_a$ is assumed to be an auxiliary variable occurring only in \eqref{eq:positive}). 
In the rest, when we say that a system has a positive variable $x$, we assume that its positiveness is guaranteed in this way. 

Let us build a sub-system called {\it $2^{i+1}$-adder (to the lower bound of a variable)}.
Based on a variable $x$ with a lower bound $n$, i.e., $x \ge n$, it aims at creating another variable with a lower bound $n+2^{i+1}$. 
Using $5i{+}5$ positive integer variables $z_1, z_2, x_0, x_b, x_c, A_k, A_k', A_k'', B_k', B_k'' (1 \le k \le i)$, we design the $2^{i+1}$-adder as follows: for $2 \le j \le i$, 
\begin{equation}\label{eq:ChenDotySeki2011}
\begin{array}{llll}
	A_1' + B_1' + x_0, 	& A_1'' + B_1'' + x_0 		&\ge& \tau, \\
	A_1  + B_1' + x_0, 	& A_1' + B_1'' + x_0 		&<& \tau, \\
	A_{j-1} + B_j' + A_j', 	& A_{j-1} + B_j'' + A_j'' 	&\ge& \tau, \\
	A_{j-1}'' + B_j' + A_j, & A_j' + B_j'' + A_{j-1}'' 	&<& \tau, \\
	A_i'' + x_b, 		& z_1 + x_c 			&<& \tau, \\ 
	z_1 + x_b, 		& z_2 + x_c 			&\ge& \tau. 
\end{array}
\end{equation}
This is actually a modification of a system of inequalities proposed in \cite{ChenDotySeki2011}, and as shown there all the inequalities of \eqref{eq:ChenDotySeki2011} but those on the last two lines imply 
\begin{equation}\label{eq:recurrence_relation_on_Ai}
	A_i'' \ge A_i + 2^{i+1}-2 
\end{equation}
for $1 \le k \le i$. 
The remaining four inequalities yield $z_1 \ge A_i''+1$ and $z_2 \ge z_1 + 1$ (in fact, they implement a $2^1$-adder). 
As a result, $z_2 \ge A_i + 2^{i+1}$. 

Assume that we have a variable $x$ whose value is at least a positive integer $n$. 
By setting $x = A_i$, we can have $A_i$ in the $2^i$-adder assume not only be positive but be at least $n_1$. 
Then $z_2 \ge n + 2^{i+1}$. 
Combining \eqref{eq:recurrence_relation_on_Ai}, $z_1 \ge A_i''+1$, $z_1 + x_c < \tau$, $A_i \ge n$, and $x_b, x_c \ge 1$ deduces that $\tau$ must be at least $n + 2^{i+1}+1$. 
It is important that for any $\tau \ge n+2^{i+1}+1$, this system can be solved as follows: 
\[
\begin{array}{l}
	A_i = n, x_0 = A_1 = \cdots = A_{i-1} = 1, \\
	A_k' = 2^k, B_k' = \tau - A_k' -1, A_k'' = 2^{k+1}-1, B_k'' = \tau - A_k'' -1 \mbox{ for $1 \le k < i$}, \\
	A_i' = n + 2^i -1, B_i' = \tau - A_i' -1, A_i'' = n + 2^{i+1}-2, B_i'' = \tau - A_i'' -1, \\
	z_1 = n+2^{i+1}-1, z_2 = n+2^{i+1}, x_b = \tau - (n+2^{i+1}-1), x_c = \tau - (n+2^{i+1}). 
\end{array}
\]
With the fact that any positive number $m \ge 1$ can be written as a sum of powers of 2, this property makes possible to combine multiple copies of $2^{i+1}$-adders inductively in order to provide a variable with an arbitrarily large lower bound. 
The number of variables involved in the system thus built is at most $\sum_{i = 1}^{\lceil \log m \rceil} (5i+5)$, which is $O((\log m)^2)$. 

Concerning \eqref{eq:ChenDotySeki2011}, the $2^{i+1}$-adder, there are two things to be noted. 
The first is that it consists of $\tau$-inequalities of at most three terms. 
The second is that we can divide its variables into four disjoint sets $V_1, V_2, V_3, V_4$ such that each inequality contains at most one variable from each of the four sets. 
One such division is: $V_1 = \{x_0, z_1, z_2\} \cup \{A_{2k}, A_{2k}', A_{2k}'' \mid 1 \le k \le \lfloor i/2 \rfloor\}$, $V_2 = \{A_{2k-1}, A_{2k-1}', A_{2k-1}'' \mid 1 \le k \le \lfloor i/2 \rfloor\}$, $V_3 = \{B_j, B_j', B_j'' \mid 1 \le j \le i\} \cup \{x_b, x_c\}$, and $V_4 = \emptyset$. 
We say that a system of $\tau$-inequalities is {\it quadripartite} if 
\begin{enumerate}
\item	every inequality in it consists of at most {\it four} terms; 
\item	its variable set can be partitioned into four disjoint subsets so that distinct terms of each inequality belong to distinct subsets. 
\end{enumerate}
(The first condition follows from the second and hence not necessary.)
Quadripartite systems of $\tau$-inequalities will play an essential role in Section~\ref{sec:complexity}. 

	\subsection{$\findoptstrength$ is \NP-hard}
	\label{subsec:findoptstrength_NPhard}

Let us prove the \NP-hardness of $\findoptstrength$. 
For the reduction, we employ a variant of 3-$\SAT$ called monotone $\oneinthree$ introduced by Schaefer \cite{Schaefer1978}, in which no literal is negated (monotonicity) and one is required to find a truth assignment such that each clause has {\it exactly one} true literal.
He proved its \NP-completeness. 
We propose its restricted variant called {\it quadripartite $\oneinthree$}, whose instance consists of a variable set that is a union of four pairwise disjoint sets $U_1, U_2, U_3, U_4$ and clauses that contain at most one variable from each of these four subsets. 

\begin{lemma}\label{lem:quad_1in3SAT}
	Quadripartite $\oneinthree$ is \NP-complete.
\end{lemma}
\begin{proof}
	Let us denote a given instance of monotone $\oneinthree$ by a pair of a set $U$ of Boolean variables and a set of clauses of three literals, which are positive due to the monotonicity. 
	We will show a polynomial-time reduction from this to an instance of quadripartite $\oneinthree$. 

	The reduction first transforms each clause of the given $\oneinthree$ instance into a quadripartite conjunction of clauses while preserving the $\oneinthree$ satisfiability, and then conjuncts them. 
	Each conjunction is designed so that it admits a partition of its variables into four disjoint sets one of which contains all variables of its source clause. 
	The $i$-th clause $\{x, y, z\}$ is converted into the following conjunction of 13 clauses: 
	\begin{equation}\label{eq:quad_1in3SAT_conjunction}
	\begin{array}{l}
		\{\neg x, a_1, b_1\} \{\neg y, a_2, b_2\} \{\neg z, a_3, b_3\} \{a_1, a_2, a_3\} \\ 
		\hspace*{15mm} \{\neg x, c, d_x\} \{\neg y, c, d_y\} \{\neg z, c, d_z\} \{c, e_{xy}. f_{xy}\} \{c, e_{yz}. f_{yz}\} \\ 
		\hspace*{30mm} \{c, e_{zx}. f_{zx}\} \{d_x, d_y, e_{xy}\} \{d_y, d_z, e_{yz}\} \{d_z, d_x, e_{zx}\}, 
	\end{array}
	\end{equation}
	where all variables but $\neg x, \neg y, \neg z$ are introduced exclusively for this clause. 
	This transformation introduces the negated literals $\neg x, \neg y, \neg z$ but they do not cause any problem because after all clauses being converted in this manner, no positive literals in the given $\oneinthree$ instance remains. 

	We claim that this preserves the $\oneinthree$ satisfiability. 
	If the $i$-th clause is $\oneinthree$ satisfied, then due to the symmetry among $x, y, z$ in \eqref{eq:quad_1in3SAT_conjunction}, it suffices to examine the case $x = 1, y = z = 0$. 
	Then this conjunction is $\oneinthree$ satisfied by setting $a_1, d_x, e_{yz}, f_{xy}, f_{zx}$ to be 1 and the others to be 0. 
	If $x = y = z = 0$, then the first three clauses in \eqref{eq:quad_1in3SAT_conjunction} imply $a_1 = a_2 = a_3 = 0$, but then the fourth clause cannot be satisfied. 
	In the case $x = y = 1$, we consider two cases depending on the value of $c$. 
	If $c = 0$, then $d_x = d_y = 1$ must hold in order to satisfy $\{\neg x, c, d_x\}$ and $\{\neg y, c, d_y\}$, but then the clause $\{d_x, d_y, e_{xy}\}$ contains too many true variables. 
	Otherwise, $d_x = d_y = e_{xy} = 0$, but then the clause cannot be satisfied. 
	Therefore, the clause $\{x, y, z\}$ is $\oneinthree$ satisfiable if and only if so are all clauses in \eqref{eq:quad_1in3SAT_conjunction}. 

	The conjunction \eqref{eq:quad_1in3SAT_conjunction} is actually not quadripartite yet, and hence, needs further transformation. 
	We replace its last clause $\{d_z, d_x, e_{zx}\}$ using the conversion 
	\begin{equation}\label{eq:quad_1in3SAT_conversion}
		\{\alpha, \beta, \gamma\} = \{\neg \alpha, h, k\}\{\neg \beta, i, k\}\{\neg \gamma, j, k\}\{h, i, j\}, 
	\end{equation}
	where $h, i, j, k$ are auxiliary variables introduced exclusively for this conversion. 
	This conversion preserves the $\oneinthree$ satisfiability and quadripartite property\footnote{
		We have not applied this (simple) conversion directly to the $i$-th clause. 
		This is because the variables involved in the conjunction thus obtained cannot be divided into four sets such that one of them contains all of $\neg x, \neg y, \neg z$. 
	}. 
	A problem is that the resulting conjunction contains both positive and negative literals of $d_z, d_x, e_{zx}$. 
	Thus, to each of the four clauses that replaced $\{d_z, d_x, e_{zx}\}$, the same conversion must apply further. 
	In this way, we finally obtain a conjunction of 28 clauses that admits quartering its variables as follows:  
	\begin{eqnarray*}
		U_{i,1} &=& \{\neg x, \neg y, \neg z, e_{xy}, e_{yz}, e_{zx}\}, \\
		U_{i,2} &=& \{a_1, b_3, c\}, \\
		U_{i,3} &=& \{a_2, b_1, d_x, d_z, f_{xy}, f_{yz}, f_{zx}\}, \\
		U_{i,4} &=& \{a_3, b_2, d_y\},
	\end{eqnarray*}
	where the variables introduced via the conversion of last clause are omitted. 
	Observe that the (negated) literals of all variables in the given SAT instance is in the same set $U_{i, 1}$, and that $(U_{i, 2} \cup U_{i, 3} \cup U_{i, 4}) \cap (U_{j, 2} \cup U_{j, 3} \cup U_{j, 4}) = \emptyset$ for any $1 \le i < j \le n$. 
	These two properties immediately bring us the quartering of variables occurring in the resulting conjunction as $U_k = \bigcup_{1 \le i \le n} U_{i, k}$ for each $k \in \{1, 2, 3, 4\}$. 
\end{proof}

\begin{theorem}\label{thm:4TP_NP-complete}
	For any $\tau \ge 4$, $\tau$-$\TP(4, 3)$ is \NP-complete. 
\end{theorem}
\begin{proof}
	A proof for $\tau=4$ comes first.
        An instance of quadripartite\footnote{The quadripartite property is not needed here, but will be so in Section~\ref{sec:complexity}.} $\oneinthree$, whose \NP-completeness was proved in Lemma~\ref{lem:quad_1in3SAT}, will be reduced into an instance of $\tau$-TP(4,~3). 
	Let us represent this instance as a pair of a set of Boolean variables $U = \{u_1, \ldots, u_n\}$ and that of clauses $C = \{c_1, \ldots, c_m\}$. 

        Let us convert this SAT instance into a system $\mathcal{S}$ of $\tau$-inequalities with positive integer variables $v_1, v_2, \ldots, v_n$ (needless to say, \eqref{eq:positive} is used here for their positiveness), which correspond to the SAT variables in $U$, such that the SAT instance is satisfiable if and only if the system is solvable. 
        In $\mathcal{S}$, the $j$-th clause of $C$, $c_j = \{u_{j_1}, u_{j_2}, u_{j_3}\}$ with $1 \le j_1, j_2, j_3 \le n$, is represented as 
	\begin{equation}\label{eq:clause_4}
		v_{j_1} + v_{j_2} + v_{j_3} \ge 4, v_{j_1} + v_{j_2} < 4, v_{j_1} + v_{j_3} < 4, \mbox{ and } v_{j_2} + v_{j_3} < 4, 
	\end{equation}
	which is equivalent to the equation $v_{j_1} + v_{j_2} + v_{j_3} = 4$ due to the assumption that $v_{j_1}, v_{j_2}, v_{j_3} \ge 1$. 
	Its solution must be that exactly one of the three variables is 2 and the others are 1.
	Therefore, if $\mathcal{S}$ is solvable, then by interpreting those in $v_1, \ldots, v_n$ with value 2 be positive and the other (that is, with value 1) negative, we can retrieve a way to satisfy the $\oneinthree$ instance; and vice versa. 
	Thus, 4-TP(4, 3) is \NP-complete (in fact, we proved that even 4-TP(3, 2) is so, see \eqref{eq:clause_4}).

	Now the result is generalized for an arbitrary $\tau \ge 4$.
	In $\mathcal{S}$, we first embed systems of $\tau$-inequalities of at most three terms, presented in Section~\ref{subsec:TP_preliminaries}, that provide auxiliary variables $x_1, x_2, x_{\tau{-}4}$ with lower bounds 1, 2, and $\tau{-}4$, respectively.
	Combining them with an inequality 
	\begin{equation}\label{eq:tau-4}
		x_1 + x_2 + x_{\tau{-}4} < \tau
	\end{equation} 
	yields $x_1 = 1$, $x_2 = 2$, and $x_{\tau{-}4} = \tau{-}4$. 
	This means that by \eqref{eq:tau-4} being embedded, in any solution of $\mathcal{S}$, these variables must admit these respective values.  
	Especially, the ``constant'' $x_{\tau{-}4}$ is added to the inequalities in \eqref{eq:clause_4}, and we obtain 
	\begin{eqnarray}
		v_{j_1} + v_{j_2} + v_{j_3} + x_{\tau{-}4} &\ge& \tau, \label{eq:clause_atleast_tau} \\
		v_{j_1} + v_{j_2} + x_{\tau{-}4} &<& \tau, \label{eq:clause_lessthan_tau}
	\end{eqnarray}
	and the analogs of \eqref{eq:clause_lessthan_tau} for $v_{j_1} + v_{j_3}$ and $v_{j_2} + v_{j_3}$. 
	Since $x_{\tau-4} = \tau-4$, these four $\tau$-inequalities are equivalent to $v_{j_1} + v_{j_2} + v_{j_3} = 4$.
	We conclude this proof by noting that any $\tau$-inequality used is either a $\ge_\tau$-inequality of at most 4 terms or a $<_\tau$-inequality of at most 3 terms, and the resulting system of $\tau$-inequalities is quadripartite. 
	This quadripartite property will become critical in proving Theorem~\ref{thm:TC_above_4_NPhard}. 
\end{proof}

Theorem~\ref{thm:4TP_NP-complete} leads us to the \NP-hardness of TP. 
Deleting the constant term $x_{\tau-4}$ from the inequalities \eqref{eq:tau-4}, \eqref{eq:clause_atleast_tau}, and \eqref{eq:clause_lessthan_tau} yields the following inequalities:  
\begin{equation}\label{eq:TP_x1x2}
	x_1 + x_2 < \tau, 
\end{equation}
\begin{equation}\label{eq:TP_jth_clause_lessthan_tau}
	v_{j_1} + v_{j_2} + v_{j_3} \ge \tau, v_{j_1} + v_{j_2} < \tau, v_{j_1} + v_{j_3} < \tau, \mbox{ and } v_{j_2} + v_{j_3} < \tau. 
\end{equation}
Consider optimizing $\tau$ subject to these. 
Due to \eqref{eq:TP_x1x2}, the minimal possible value of $\tau$ is 4. 
Then its optimal value is 4 if and only if the $\oneinthree$ instance is satisfiable. 

\begin{lemma}\label{lem:TP43_NPhard}
	$\TP(4, 3)$ is \NP-hard. 
\end{lemma}

\begin{remark}\label{rmk:strictly_lessthan_tau}
	All the systems of $\tau$-inequalities designed in this section so far can be solved even subject to extra condition that all variables being strictly less than $\tau$. 
	This extra condition does not prevent the systems from playing their intended roles when being embedded. 
\end{remark}

Using this instance, now we can prove that $\findoptstrength$ is \NP-hard. 
Making use of the fact that all of its $\tau$-inequalities contain at most 4 terms, we transform them into cooperation sets of a strength-free TAS. 
The inequalities \eqref{eq:TP_jth_clause_lessthan_tau}, which are for the clause $c_j$, are encoded as $t_{c_j} = (v_{j_1}, v_{j_2}, v_{j_3}, x_\tau)$ with $\mathcal{D}(t_{c_j}) = \mathcal{P}(\{\north, \west, \south, \east\}) \setminus \{\{\north\}, \{\west\}, \{\south\}, \{\north, \west\}, \{\north, \south\}, \{\west, \south\}\}$. 
Note that its east side is labeled with an auxiliary variable $x_\tau$, but this variable does not play any essential role but filling the blank side (the tile is not triangle but square). 
Likewise, the inequalities that aim at forcing $v_1, \ldots, v_n$ be positive are converted. 
The inequalities in \eqref{eq:positive} are encoded as a tile type $t_i = (x_1, x_a, x_\tau, x_\tau)$ with $\mathcal{D}(t_i) = \mathcal{P}(\{\north, \west, \south, \east\}) \setminus \{\{\north\}, \{\west\}\}$. 
Though being not mentioned in \eqref{eq:positive}, we can assume $x_1 < \tau$ as noted in Remark~\ref{rmk:strictly_lessthan_tau}. 
The inequality \eqref{eq:TP_x1x2} is encoded as $t' = (x_1, x_2, x_\tau, x_\tau)$ with $\mathcal{D}(t') = \mathcal{P}(\{\north, \west, \south, \east\}) \setminus \{\{\north\}, \{\west\}, \{\north, \west\}\}$. 
Then, there exists a TAS at temperature 4 that is locally equivalent to this strength-free TAS if and only if the given instance of $\oneinthree$ is satisfiable.

\begin{theorem}\label{thm:findoptstrength_NPhard}
	$\findoptstrength$ is \NP-hard. 
\end{theorem}

	\section{Tile complexity at high temperatures}
	\label{sec:complexity}

In Section~\ref{sec:behavior}, bridges from $\oneinthree$ to TP(4, 3) and further to the {\it local}, or {\it micro}, behavior of a TAS have been established.  
Since the study of TAS aims at facilitating the design of nano-scale structures, we should shift our focus onto their {\it global} (or {\it macro, terminal}) behavior; how a given shape is built by TASs. 
Problems of interest include: for any temperature $\tau$ given as parameter, 
\begin{enumerate}
\item	Is there a shape $S_\tau$ that prefers the temperatures above $\tau$ to the lower ones in terms of tile complexity; that is, $\dtilec{< \tau}(S_\tau) > \dtilec{\tau}(S_\tau)$? 
	If so, then can we design an algorithm to construct $S_\tau$? 
\item	Can we compute the directed tile complexity of a shape $S$ at the temperatures below $\tau$, that is, $\dtilec{\le \tau}(S)$, in a polynomial time? 
\end{enumerate}
For $\tau = 2$, these problems have been studied intensively. 
Adleman et al.~proved that $\dtilec{2}(\square{n}) = O(\frac{\log n}{\log \log n})$ for any $n \times n$ square $\square{n}$ \cite{AdChGoHu2001}. 
In contrast, $\dtilec{1}(\square{n}) \le 2n-1$ and this bound is conjectured to be tight \cite{AdChGoHuKeEsRo2002}, which is highly probable. 
Thus, $\dtilec{< 2}(\square{n}) \gg \dtilec{2}(\square{n})$, provided the conjecture is true. 
As for the second problem, it is \NP-hard to compute the directed tile complexity at the temperatures at most 2 \cite{AdChGoHuKeEsRo2002}. 

We will work on these problems without any constraint on temperature, and answer them by proving the following two theorems. 

\begin{theorem}\label{thm:tau_better_than_less}
	For any $\tau \ge 2$, there is a shape $S_\tau$ whose directed tile complexity is strictly lower at the temperature $\tau$ than at any temperatures below $\tau-1$. 
\end{theorem}

\begin{theorem}\label{thm:TC_above_4_NPhard}
	For any $\tau \ge 4$, it is \NP-hard to compute the directed tile complexity of a shape at the temperatures below $\tau$. 
\end{theorem}

Theorem~\ref{thm:tau_better_than_less} is a positive answer to the first problem. 
Based on the results obtained on TP(4,~3), we propose a design of a shape $S_\tau$ for a given $\tau \ge 2$ that prefers the temperature $\tau$ to the ones strictly below in terms of tile complexity; that is, $\dtilec{< \tau}(S_\tau) > \dtilec{\tau}(S_\tau)$. 
This should be, as of now, interpreted as an infeasibility result that the arbitrarily fine control of binding energies is necessitated to design TASs of ``reasonable size'' for a certain shape in a laboratory, which has not yet been realized, at least to my knowledge. 

Not only to this end but this design also makes possible to convert the instance of $\oneinthree$ into a shape $S$ and a constant $c$ such that the instance is satisfiable if and only if there is a directed TAS of at most $c$ tile types that strictly self-assembles $S$ at a temperature below $\tau$ for an arbitrary $\tau \ge 4$; that is, $\dtilec{\le \tau}(S) \le c$. 
This amounts to the proof of Theorem~\ref{thm:TC_above_4_NPhard}, which answers the second problem unless ${\tt P} = \NP$. 
Adleman et al.~proved the analogous result for $\tau = 2$ \cite{AdChGoHuKeEsRo2002}. 
The case $\tau = 3$ remains open. 
This gap must be filled, but our proof cannot be applied to this case, at least directly. 
It is probably essential to transform an instance of 3-$\SAT$ into finely-crafted gadgets for this case as done by Adleman et al.~in \cite{AdChGoHuKeEsRo2002} for the case $\tau = 2$. 
This paper leaves this case open. 

\begin{figure}[tb]
\begin{center}
\includegraphics[scale=0.6]{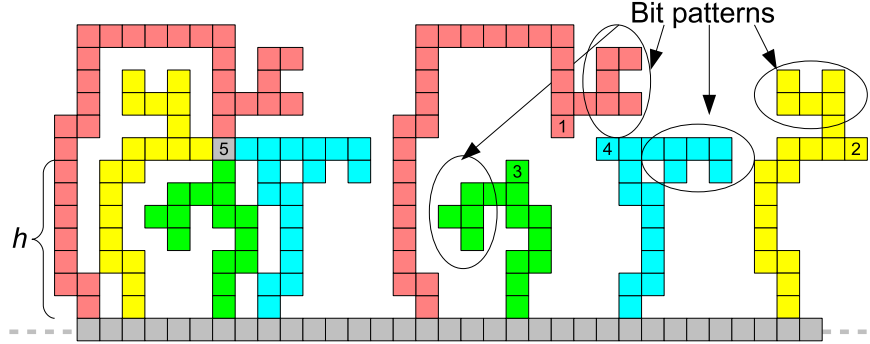}
\end{center}
\caption{
	A logical component for a $\ge_\tau$-inequality of four terms (left) and four variable trees (red, green, blue, and yellow from the left) that correspond to the four variables in the inequality. 
	The positions numbered 1, 2, 3, and 4 are the cooperation tip of these trees, respectively. 
	The logical component consists of a single gray tile at the position numbered 5 and four pillars that are of the shape identical to the four respective variable trees. 
}
\label{fig:4vars_atleast_tau}
\end{figure}

We propose the following unified approach to various problems related to the behaviors and temperatures of directed TASs including the above two problems: 
\begin{enumerate}[Step 1.]
\item	choose a quadripartite system $\mathcal{S}$ of $\tau$-inequalities properly for the purpose among those built in Section~\ref{sec:behavior}; hence, all $\ge_\tau$-inequalities of $\mathcal{S}$ are of at most 4 terms and all of its $<_\tau$-inequalities are of at most 3 terms; 
\item	convert its variables $v_1, v_2, \ldots, v_n$ into trees of height $h$ and distinct shape, which we call {\it variable trees}, where $h$ is a parameter adjustable for our convenience;
\item	for each $\tau$-inequality, bundle the (at most 4) trees thus converted from its variables into a shape called the {\it logical component}; 
\item	(two copies of) the $n$ variable trees and logical components each are mounted next to each other onto a scaffold, and amounts to a shape, which we denote by $S$. 
\end{enumerate}
Examples of variable trees and a logical component are shown in Figure~\ref{fig:4vars_atleast_tau}. 
Note that $S$ is parameterized by $h$. 
The design of the shape $S$ through these four steps will bring a constant $c \ll h$ with the following property after its parameter $h$ being made large enough. 

\begin{property}\label{property:TAS_as_inequality_system_solver}
Any directed TAS $\tas = (T, \sigma, g, \tau)$ that strictly self-assembles $S$ using at most $(n + c')h + c$ tile types has (not-necessarily-distinct) $n$ glue labels $\ell_1, \ell_2, \ldots, \ell_n$ such that 
\begin{itemize}
\item	for any $m \le 4$ and $k_1, \ldots, k_m \in \{1, \ldots, n\}$, if $\mathcal{S}$ includes the $\ge_\tau$-inequality $\sum_{1 \le i \le m} v_{k_i} \ge \tau$, then $\sum_{1 \le i \le m} g(\ell_{k_i}) \ge \tau$; 
\item	for any $m' \le 3$ and $k_1, \ldots, k_{m'} \in \{1, \ldots, n\}$, if $\mathcal{S}$ includes the $<_\tau$-inequality $\sum_{1 \le j \le m'} v_{k_j} < \tau$, then $\sum_{1 \le j \le m'} g(\ell_{k_j}) < \tau$, 
\end{itemize}
where $c'$ is the number of $<_\tau$-inequalities in the system $\mathcal{S}$ chosen at Step 1. 
\end{property}

This property should be interpreted that the way for the function $g$ of such a ``small'' $\tas$ at the temperature $\tau$, if any, to assign labels with glue strengths tells how to satisfy the given system $\mathcal{S}$ of $\tau$-inequality. 
Conversely speaking, unless the given system admits a solution with $\tau$ being some specific value $\tau'$, any directed TAS at the temperature $\tau'$ needs strictly more than $(n+c')h+c$ tile types\footnote{
Actually, it will be proved necessary for $\tas$ to have at least $(n+c'+1)h$ tile types, which is much larger than $(n+c')h+c$ because $h \gg c_2$. 
} in order to strictly self-assemble $S$. 
Verifying this property amounts to proofs of Theorems~\ref{thm:tau_better_than_less} and \ref{thm:TC_above_4_NPhard}; it is sufficient for us to choose a proper $\mathcal{S}$ at Step 1. 
In order to prove Theorem~\ref{thm:tau_better_than_less}, for $k \ge 2$, we choose the system of $\tau$-inequalities that cannot be solved for any $\tau < k$ but can for any $\tau \ge k$. 
This is the one designed in Section~\ref{subsec:TP_preliminaries} based on the $2^{i+1}$-adders. 
It is transformed into a shape $S$ through Steps 2-4. 
This amounts to a proof of a slightly-stronger version of Theorem~\ref{thm:tau_better_than_less}. 
As for Theorem~\ref{thm:TC_above_4_NPhard}, the TP instance built for Theorem~\ref{thm:4TP_NP-complete} is rather chosen. 
In the rest of this section, therefore, we just have to verify Property~\ref{property:TAS_as_inequality_system_solver} in order to complete the proof of these theorems. 

Let $V = \{v_1, v_2, \ldots, v_n\}$ be the set of variables occurring in the system $\mathcal{S}$ chosen at Step 1. 
As mentioned in Remark~\ref{rmk:strictly_lessthan_tau}, we can assume that every variable in $V$ is {\it strictly} less than $\tau$; this assumption simplifies the design of $S$ and our explanation below. 
The more essential is that $\mathcal{S}$ is quadripartite, that is to say, the variable set $V$ can be partitioned into four {\it disjoint} subsets $V_\north, V_\west, V_\south, V_\east$. 

\subsubsection*{Step 2}

In this step, we convert each variable in $V$ into one of the four tree shapes illustrated in Figure~\ref{fig:4vars_atleast_tau} depending on which of $V_\north, V_\west, V_\south, V_\east$ it belongs to. 
For instance, each variable in $V_\north$ is converted into the shape whose tip is numbered 1 (and colored red). 
This shape is a tree of height $h$ plus some constant, not depending on $h$, with only one crotch for two branches; one consists of a bit pattern and the other is of size 1 for cooperation called {\it cooperation tip}, which is numbered 1 in Figure~\ref{fig:4vars_atleast_tau}. 
We say that this shape is of {\it north type} after the subscript of $V_\north$. 
Shapes of north type that are thus converted from distinct variables (in $V_\north$) are identical mod their bit patterns of length $\lceil \log n \rceil$. 
In this way, each variable in the other three variable subsets $V_\west, V_\south, V_\east$ is also converted into the shapes that are numbered 2, 3, 4 (and colored yellow, green, blue) in Figure~\ref{fig:4vars_atleast_tau}, respectively, and we say that these shapes are respectively of {\it west, south, and east type}. 
All the $n$ variables of $V$ have been now associated with $n$ different tree shapes of proper type. 
We collectively refer to them as {\it variable trees}. 
Two copies of each variable tree is mounted onto a scaffold. 

Every conversion of $<_\tau$-inequalities in $\mathcal{S}$ into logical components at Step~3 needs to introduce an auxiliary tree of the north, west, south, or east type and recall that $\mathcal{S}$ was assumed to contain $c'$ of them (the conversion of $\ge_\tau$-inequalities does not). 
These trees are distinguished from each other and from the variable trees by their bit patterns. 
Two copies of each auxiliary tree are mounted next to each other on the scaffold of $S$. 
We will call variable trees and auxiliary trees simply trees unless confusion arises. 

We claim that $\tas$ needs at least $(n+c')h$ tile types for $\tas$ to strictly self-assemble these $n+c'$ trees mainly using the results by Adleman et al.~in \cite{AdChGoHuKeEsRo2002} on tile complexity of trees. 
Being duplicated, each tree has its copy that assembles upward from the scaffold (actually, all trees but at most one do so since $\tas$ is assumed to be singly-seeded). 
Theorem~4.3 in \cite{AdChGoHuKeEsRo2002} implies that in order for $\tas$ to strictly self-assemble two variable trees of different types (e.g., north and west), at least $2h$ plus some constant number of tile types are necessary; $h$ exclusively for each, no matter how and at what temperature $\tas$ is designed, where the constant is much smaller than $h$. 
This is not to say that $\tas$ could not reuse any tile type for both of them. 
It can, for example, for their bit patterns. 
However, the number of such reusable tile types is bounded by a constant that is much smaller than $h$. 
This lower bound exists also for two trees of identical type but with distinct bit patterns. 
It is not allowed for $\tas$ to put tiles of identical type at the crotch positions because if they were the same, incorrect bit patterns could assemble.
In contrast, $\tas$ is able to reuse tiles of same type for the cooperation tips of these trees, but anyway it saves only one tile type. 
This observation is immediately generalized for more variable and auxiliary trees as: in order for $\tas$ to assemble the $n$ distinct variable trees and $c'$ auxiliary trees, at least $(n+c')h + c''$ tile types are necessary for some constant $c'' \ll h$. 
In brief, 
\[
	\dtilec{\tau}(S) \ge (n+c')h + c''
\]
at {\it any temperature} $\tau$. 

Before proceeding to the explanation of next step, we prove that $\tas$ must put the same tile type at the cooperation tip positions of two copies of each variable tree or auxiliary tree. 
This is related to the position of the seed of $\tas$ so that first we see that it merely costs the number of tile types for $\tas$ to put its head on a variable tree or an auxiliary tree. 
Suppose $\tas$ put its seed on one copy of some variable tree, and let us compare it with the other copy of the same tree. 
Recall that they are located next to each other. 
In order to start assembling the seed-free copy, $\tas$ first must assemble downward from the seed and the scaffold between the copies. 
Then consider the process for $\tas$ to assemble the seed-free copy upward up to the counterpart position of the seed.
During this process, $\tas$ cannot reuse any tile type used for the part between the seed and scaffold; such a recycle would enable $\tas$ to skip the assembly of the seed-free copy and produce a shape with only one copy of the variable tree. 
Thus, even if the seed is on a copy of some variable tree, it could not be located ``far from'' the scaffold if it were not for extra $n$ tile types available for $\tas$. 
This means that $\tas$ needs to assemble almost all parts of these trees upward, and hence, it would impose the unaffordable number of extra tile types on $\tas$ to put tiles of distinct types at the cooperation tips of two copies of a tree. 
Now, for $1 \le i \le n$, we can denote {\it the} tile type that $\tas$ puts at the both cooperation tips of the copies of the variable tree for $v_i$ by $t_i$. 
The label $\ell_i$ mentioned in Property~\ref{property:TAS_as_inequality_system_solver} is found at the side of $t_i$ opposite to the crotch. 
That is, 
\[
	\ell_i = 
	\begin{cases}
	t_i(\south) & \text{if $v_i \in V_\north$} \\
	t_i(\east)  & \text{if $v_i \in V_\west$} \\
	t_i(\north) & \text{if $v_i \in V_\south$} \\
	t_i(\west)  & \text{if $v_i \in V_\east$}. 
	\end{cases}
\]

\subsubsection*{Step 3}

In this step, we convert $\ge_\tau$-inequalities and $<_\tau$-inequalities of the given system $\mathcal{S}$ into shapes which we call {\it logical components}. 
Two of their copies will be mounted in Step 4 onto the scaffold, which will complete the design of the shape $S$. 

Let us begin with a simpler one: the logical component for $\ge_\tau$-inequality (of at most 4 terms). 
Recall that any $\ge_\tau$-inequality in $\mathcal{S}$ is of at most 4 terms and none of its two terms are variables taken from the same variable subset. 
That is to say, the variable trees that its terms (variables) correspond to are of pairwise distinct type, and hence, can be bundled together so as for their cooperation tips to be adjacent to one position (see Figure~\ref{fig:4vars_atleast_tau}, where the position is numbered 5).
The logical component for this $\ge_\tau$-inequality consists of the variable trees thus positioned, which we refer to as {\it pillars}, and the position adjacent to all the cooperation tips of the pillars. 
In this manner, each $\ge_\tau$-inequality is converted into a logical component and its two copies are mounted onto the scaffold of $S$. 

\begin{description}

\item[{\it Proof of the first item of Property~\ref{property:TAS_as_inequality_system_solver}}] \ \\
Now we will prove the first item of Property~\ref{property:TAS_as_inequality_system_solver}, that is, if the $\ge_\tau$-inequality 
\begin{equation}\label{eq:at_least_tau_inequality}
	\sum_{1 \le i \le m} v_{k_i} \ge \tau
\end{equation}
belongs to $\mathcal{S}$, then 
\begin{equation}\label{eq:at_least_tau_attachment}
	\sum_{1 \le i \le m} g(\ell_{k_i}) \ge \tau,
\end{equation} 
where $m \le 4$ and $k_1, \ldots, k_m \in \{1, \ldots, n\}$. 
Before that, however, let us quickly show that if \eqref{eq:at_least_tau_attachment} holds, then $\tas$ can assemble the logical component for \eqref{eq:at_least_tau_inequality} by introducing only one new tile type in the following manner. 
$\tas$ first assembles the respective parts of the component that correspond to the variable trees for $v_{k_1}, \ldots, v_{k_m}$ as done for the copies of these trees (hence, costs nothing) and lets their cooperation tips to cooperate to attach a tile of the new type to fill the position surrounded by the $m$ cooperation tips, which is numbered 5 in the logical component in Figure~\ref{fig:4vars_atleast_tau}.
This cooperation provides enough strength for stable tile attachment due to \eqref{eq:at_least_tau_attachment}. 
Let us return to the proof. 
For the sake of contradiction, suppose \eqref{eq:at_least_tau_attachment} does not hold. 
Focus on the seed-free copy of the logical component for \eqref{eq:at_least_tau_inequality} for the concise argumentation. 
Since this copy is free from seed, $\tas$ needs to assemble at least one pillar of the component upward from the scaffold. 
Under the supposition that \eqref{eq:at_least_tau_attachment} did not hold, $\tas$ is not allowed to assemble these upward-growing pillars as done for the corresponding variable trees without any extra tile types. 
More precisely, at the cooperation tip of at least one of these pillars, say the one for $v_{k_i}$, $\tas$ must put a tile of different type from $t_{k_i}$. 
However, this costs extra at least $h$ tile types because $\tau$ cannot reuse any tile types for the tree for $v_{k_i}$ in order to assemble this pillar, or any tile types for the any other tree $v_{k_j}$ because $\tas$ was proved not to be able to assemble $v_{k_i}$ and $v_{k_j}$ using strictly less than $2h$ tile types (this argument works in order to prove that this pillar cannot assemble using tile types for auxiliary trees). 
\qed
\end{description}



\begin{figure}[tb]
\begin{center}
\includegraphics[scale=0.55]{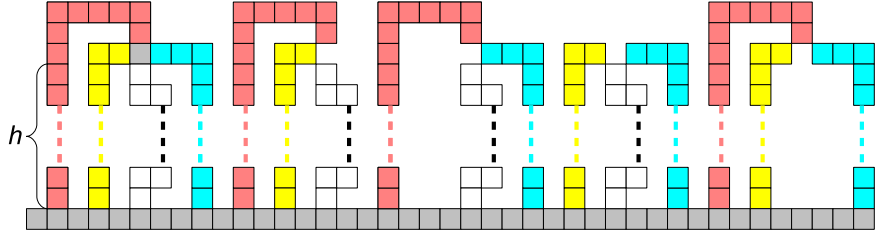}
\end{center}
\caption{
	A logical component for a $<_\tau$-inequality (of 3 terms). 
	The trees of north, west, east type are variable trees (red, yellow, blue), while the remaining tree of south type is an auxiliary one (white). 
	The trees are written concisely for clarity, but they should be of the same shape as illustrated in Figure~\ref{fig:4vars_atleast_tau}. 
}
\label{fig:3vars_lessthan_tau}
\end{figure}

The design of logical components for $<_\tau$-inequalities (of at most 3 terms) is more involved. 
One thing to be noted first is that this cannot be realized simply by modifying the above-mentioned component design for $\ge_\tau$-inequalities by leaving the position adjacent to all the cooperation tips of the pillars empty. 
This is because no attachment may not mean the insufficient strength but label mismatch. 
We propose the following design instead. 

Let us consider a $<_\tau$-inequality:
\begin{equation}\label{eq:less_than_tau_inequality}
	\sum_{1 \le j \le m'} v_{k_j} < \tau 
\end{equation}
for some $m' \le 3$ and $k_1, \ldots, k_{m'} \in \{1, \ldots, n\}$. 
In addition to the $m'$ variable trees for $v_{k_1}, \ldots, v_{k_{m'}}$, we employ one of the $c'$ auxiliary trees prepared in Step~2 that is of distinct type from any of the variable trees (recall that $c'$ is the number of $<_\tau$-inequalities in $\mathcal{S}$). 
Then we combine copies of at least $m'$ of these $m'+1$ trees as done in building the components for $\ge_\tau$-inequalities but in all combinations of $m'$ trees. 
The combination of all the $m'+1$ trees is turned into a ``gadget'' by filling the position that are adjacent to all the cooperation tips, whereas the other combinations are considered to be gadgets without any modification. 
These $m'+2$ gadgets amount to the logical component for the $<_\tau$-inequality \eqref{eq:less_than_tau_inequality}, and two copies of this component are mounted on the scaffold of $S$. 
Note that, once being used, the auxiliary tree will not reused for the component of other $<_\tau$-inequalities any more. 

\begin{description}

\item[{\it Proof of the second item of Property~\ref{property:TAS_as_inequality_system_solver}}] \ \\
	Now we will prove that if \eqref{eq:less_than_tau_inequality} holds, then 
	\begin{equation}\label{eq:less_than_tau_attachment}
		\sum_{1 \le j \le m'} g(\ell_j) < \tau.
	\end{equation}
	Before this proof, let us note that since the auxiliary tree does not appear in any other component, one can have $g$ assign an arbitrary strength to the glue at the cooperation tip of this tree without causing any malfunctioning of the other component, as long as \eqref{eq:less_than_tau_attachment} holds. 
	Setting this strength to be $\tau - \sum_{1 \le j \le m'} v_{k_j}$ enables $\tas$ to assemble this component with introducing only one new tile type that fills the position adjacent to the cooperation tips. 

	For the sake of contradiction, suppose that the sum $\sum_{1 \le j \le m'} g(\ell_j)$ were at least $\tau$. 
	Focus on the seed-free copy of this component, and especially first on the gadget with $m'+1$ pillars, with the position surrouned by the cooperation tips being filled (the leftmost gadget in Figure~\ref{fig:3vars_lessthan_tau}). 
	Let us denote this gadget by $G_1$. 
	Due to the seed-freeness, some of the pillars of $G_1$ assemble upward and have a tile $t$ fill the position by cooperative attachment. 
	Assume that some pillar is absent from the cooperation. 
	Let us denote the gadget that does not contain the ``lazy'' pillar by $G_2$. 
	Then some of the tiles put at the cooperation tips of $G_2$ must be distinct from those at the corresponding positions of $G_1$ in type in order to prevent the tile $t$ from attaching to $G_2$. 
	However, this distinctiveness would cost $\tas$ extra $h$ new tile types, as explained in the proof of the first item. 
	Thus, all of the pillars of $G_1$ must join the cooperation, and hence, all the pillars of $G_1$ assemble upward. 
	Then consider the gadget that is free from the auxiliary pillar, which we denote by $G_3$. 
	The types of tiles at the corresponding cooperation tips of $G_1$ and $G_3$ must be the same in order for $\tas$ not to pay the unaffordable extra $h$ tile types. 
	However, then there must exist $1 \le j \le m'$ such that the tile $t_j$ at the cooperation tip of the tree for $v_j$ and the tile at the corresponding position in $G_1$ must differ from each other in type because of the supposition $\sum_{1 \le j \le m'} g(\ell_j) \ge \tau$. 
	This again causes the unaffordable cost. 
	Consequently, we can say that \eqref{eq:less_than_tau_attachment} must hold. 
	\qed
\end{description}

As the readers might have already noticed, the constant $c$ is the numbers of tile types for the scaffold (actually, equal to its length, which depends only on the number of variables and the number of inequalities in $\mathcal{S}$) plus the number of tile types for each variable tree that exceeds $h$ (top structures differ tree by tree, but a bit). 

When being introduced, the directed tile complexity at the temperatures at most $\tau$ was noted not to be a monotonically decreasing function. 
Let us conclude this section and hence this paper by its proof. 
Since we do not find this non-monotonicity significant in practice, we simply exemplify it and will not inquire deeper into the matter. 
Actually, this proof is a good opportunity to show the usefulness of the unified framework we proposed previously. 

\begin{lemma}
	There exists a shape $S$ such that $\dtilec{3}(S) < \dtilec{4}(S)$. 
\end{lemma}
\begin{proof}
	As explained into details when being introduced, the framework needs only our appropriate choice of a system of $\tau$-inequalities at Step 1. 
	Actually, it suffices to find as such a system the one that is solvable for $\tau = 3$ but not for $\tau = 4$. 
	Such an equation, however, was already given as \eqref{eq:TP_example_1} in Section~\ref{subsec:TP_preliminaries}. 
	Needless to say, all $\ge_\tau$-inequalities and all $<_\tau$-inequalities in this system are of at most 4 and 3 terms, respectively, and this system is quadripartite. 
\end{proof}

A problem of theoretical interest is whether for any temperature $\tau$, there exists a shape $S$ such that $\dtilec{k}(S)$ is strictly decreasing over the interval $[1, \tau]$. 
The framework cannot be applied, or even if it can, some modification is necessary.  

	\section{Conclusions}

In this paper, we continued research on the behavioral equivalence among TASs at the local level that was initiated in \cite{ChenDotySeki2011} and solved their open problems of whether minimizing the temperature of TASs that behave locally as specified as input can be solved in a polynomial time by proving that this problem is \NP-hard. 
We deduced this from our other result that the threshold programming, a special form of integer programming, is still \NP-hard. 
Furthermore, we proposed a unified framework to work on various problems related to the temperature and the global behavior of TASs, which makes use of systems of $\tau$-inequalities designed for the proof of the above results. 

There are several directions for further research. 
The \NP-hardness stated in Theorem~\ref{thm:findoptstrength_NPhard} never eliminates the possibility to improve the algorithm by Chen, Doty, and Seki so as to optimize the size and temperature of TASs for the $n \times n$ square $\square{n}$ simultaneously.  
Elucidating this will deepen our knowledge of the tile complexity of $\square{n}$ further. 
As for Theorem~\ref{thm:tau_better_than_less}, our interest lies on the ratio $\dtilec{< \tau}(S_\tau)/\dtilec{\tau}(S_\tau)$. 
In our construction of $S_\tau$, the ratio is constant no matter how $h$ is adjusted. 
In contrast, $\dtilec{1}(\square{n})/\dtilec{2}(\square{n})$ can get arbitrarily large by making $n$ larger, provided the conjecture $\dtilec{1}(\square{n}) = 2n-1$ is true (or even $\dtilec{1}(\square{n}) = \Omega(n)$ is fine). 
Can we find such a shape for arbitrary temperature? 
With Theorem~\ref{thm:TC_above_4_NPhard} and the result by Adleman et al.~\cite{AdChGoHuKeEsRo2002}, the \NP-hardness of computing directed tile complexity at the temperatures below 3 must be proved.

	\section*{Acknowledgements}

We gratefully acknowledge helpful discussions with David Doty, Hiro Ito, and Kei Taneishi on several points in this paper. 
The earlier versions of this paper owe much to the valuable comments from Kojiro Higuchi, Natasha Jonoska, and the anonymous referees of CiE 2012. 

This research is supported mainly by the Kyoto University Start-up Grant-in-Aid for Young Researchers No. 021530 to Shinnosuke Seki and by the Funding Program for Next Generation World-Leading Researchers (NEXT program) to Yasushi Okuno. 
Works by the first author are also financially supported in part by Helsinki Institute of Information Technology (HIIT) Pump Priming Project Grant (902184/T30606) and by Department of Information and Computer Science, Aalto University. 

	\bibliographystyle{plain}
	\bibliography{hightempTAS}

\end{document}